\documentclass{article}

\usepackage[main, final]{neurips_2026}

\makeatletter
\def\@notice{}

\def\ps@headings{%
    \def\@oddhead{\hfill \hfill} 
    \def\@evenhead{\hfill \hfill}
    \def\@oddfoot{\hfill\thepage\hfill}
    \def\@evenfoot{\hfill\thepage\hfill}}

\def\ps@empty{%
    \def\@oddhead{\hfill \hfill} 
    \def\@evenhead{\hfill \hfill}
    \def\@oddfoot{\hfill\thepage\hfill}
    \def\@evenfoot{\hfill\thepage\hfill}}
\makeatother

\usepackage[utf8]{inputenc}
\usepackage[T1]{fontenc}
\usepackage{microtype}

\usepackage{hyperref}
\usepackage{url}
\usepackage{xspace}

\usepackage{amsmath}
\usepackage{amsfonts}
\usepackage{amssymb}
\usepackage{amsthm}
\usepackage{nicefrac}

\usepackage{graphicx}
\usepackage{subcaption}

\usepackage{booktabs}
\usepackage{multirow}
\usepackage{multicol}
\usepackage{array}
\usepackage[table]{xcolor}   
\usepackage{colortbl}

\usepackage{algorithm}
\usepackage{algpseudocode}

\definecolor{tableheader}{RGB}{46,134,171}
\definecolor{tablerowalt}{RGB}{245,248,250}
\definecolor{bestresult}{RGB}{46,134,171}

\newcommand{\best}[1]{\textbf{#1}}
\newcommand{\second}[1]{\underline{#1}}
\newcommand{\MARS}{{MARS}\xspace}
\newcommand{\MARSv}{{MARS-T}\xspace}
\newcommand{\MARSm}{{MARS-M}\xspace}

\newtheorem{theorem}{Theorem}

\newtheorem{proposition}[theorem]{Proposition}

\usepackage{wrapfig}
\usepackage{float}

\title{\MARS: Multi-rate Aggregation of Recency Signals for Sequential Recommendation across Sparse and Dense Regimes}

\author{%
  Zhenyu Yu$^{1}$, Shuigeng Zhou$^{1}$\thanks{Corresponding author.} \\
  $^{1}$College of Computer Science and Artificial Intelligence, 
  Fudan University\\
  \texttt{yuzhenyuyxl@foxmail.com, sgzhou@fudan.edu.cn} 
}

\begin{document}

\maketitle

\begin{abstract}
Sequential recommenders weight historical interactions either through positional self-attention as in Transformers or through a single implicit decay schedule as in State-Space Models. Neither makes the multi-scale temporal structure of real user behaviour explicit. We propose \MARS, an encoder-agnostic aggregation operator that consumes real timestamps and produces $K$ summaries emphasising distinct recency scales, fused by a context-adaptive gate. \MARS adds at most $6\%$ parameters and runs in $\mathcal{O}(LdK)$ time. \MARS\ adapts to data density by automatically selecting between two encoder instantiations: \MARSv\ (Transformer) for sparse data and \MARSm\ (Mamba) for dense data, based on the average sequence length of the training set. On five public benchmarks against ten Transformer- and Mamba-based baselines under a unified RecBole protocol, \MARS\ attains the best HR@10 on every benchmark, with mean relative gain $+19.7\%$ over the strongest content-only Transformer baseline on sparse data (reaching $+36.2\%$ on Games) and $+3.2\%$ HR@10 / $+0.9\%$ NDCG over SIGMA on dense ML-1M at $42\%$ fewer MFLOPs, occupying the accuracy-efficiency Pareto frontier across the data-density spectrum. A backbone-only ablation isolates the marginal contribution of \MARS at $+4\%$ to $+19\%$ HR@10 on sparse data and motivates the dual-instantiation design. The code is included in the supplementary material.
\end{abstract}

\section{Introduction}
\label{sec:intro}

Sequential recommendation predicts the next item a user will interact with given their past interaction history. The temporal structure of such histories is heterogeneous; user interests evolve across multiple time-scales, from session-level bursts within minutes to long-term preferences spanning months. Capturing this multi-scale structure requires the model to combine fine-grained local patterns with coarse-grained global trends. Two families of sequential recommenders dominate the literature, Transformers~\citep{kang2018sasrec,sun2019bert4rec,du2023fearec} and State-Space Models~\citep{gu2023mamba,liu2024mamba4rec,liu2025sigma}. Both compress the multi-scale structure into a single mechanism. Transformers rely on ordinal positional encodings, with time-aware variants~\citep{li2020tisasrec} adding one global relative-time parameterisation. State-Space Models implement a single implicit decay schedule determined by their learned state-transition parameters. Neither family exposes the multiple time-scales the data exhibits, leaving a gap between the modelling assumption and the underlying data-generating process.

We trace this limitation to a conflation of two design decisions that are in principle orthogonal, the \emph{sequence encoder} that maps the items $\{x_1,\dots,x_i\}$ to a hidden state $\mathbf{h}_i$, and the \emph{time-aware aggregation} that pools the hidden states $\{\mathbf{h}_1,\dots,\mathbf{h}_L\}$ into a user representation $\mathbf{h}_u$. The dominant paradigm bundles them. The encoder simultaneously performs content mixing through dot-product attention or selective state updates and temporal weighting through positional embeddings, relative-time biases, or state-transition decays. The optimal architecture for sequence mixing need not coincide with the optimal parameterisation of multi-scale temporal weighting, so coupling the two restricts the design space and forces a single sequence encoder to absorb both responsibilities.

We propose to decouple the two responsibilities by introducing time-aware aggregation as a small learnable module placed on top of any sequence encoder. We realise this idea in \MARS, an aggregation operator that consumes the encoder's hidden states $\mathbf{H}=[\mathbf{h}_1,\dots,\mathbf{h}_L]$, interpreted as the recency signals of past interactions, together with the real timestamps $\{t_1,\dots,t_L\}$ of the events. \MARS produces $K$ summaries of $\mathbf{H}$, each emphasising a distinct recency scale through a learnable exponential decay rate $\lambda_k$, and fuses them into a final user representation through a context-adaptive gate. Three design choices distinguish \MARS from prior time-aware sequential models. The decay rates $\{\lambda_k\}$ are modulated per user via a small MLP conditioned on the sequence summary, allowing distinct users to emphasise distinct time-scales. A Jensen--Shannon diversity penalty on the per-head attention distributions prevents the $K$ heads from collapsing onto a single rate, an issue we analyse through a Hawkes-process generative model. A context-adaptive gate conditioned on the sequence content learns to route between the $K$ summaries on a per-user basis. The aggregated user vector is residually added to the encoder's last-position state, preserving the standard sequential-recommendation interface and allowing \MARS to be inserted into any existing pipeline without architectural changes downstream.

\textbf{Contributions.}
\begin{itemize}
\item \textbf{A new design axis.} We identify time-aware aggregation as a design axis separate from the sequence encoder, and propose \MARS, a backbone-agnostic operator with user-conditioned decay, JSD-regularised diversity, and an adaptive fusion gate.
\item \textbf{Theoretical guarantees.} We provide a three-part theoretical justification. $K$-mixtures of exponentials are universal approximators of monotone decay kernels (Proposition~\ref{prop:approx}). Under a Hawkes-process data-generating model, the diversity-regularised \MARS recovers the ground-truth decay rates up to permutation (Proposition~\ref{prop:identif}, the Hawkes connection is novel to this work). \MARS runs in $\mathcal{O}(LdK)$ time (Proposition~\ref{prop:complexity}).
\item \textbf{Unified framework.} We instantiate \MARS\ as a unified framework with two complementary backbones, \MARSv\ (Transformer) for sparse data and \MARSm\ (Mamba) for dense data. Together they attain the best HR@10 on all five public benchmarks under a unified evaluation protocol, with a backbone-only ablation that quantifies the regime-dependent marginal contribution of \MARS\ and validates the dual-instantiation design.
\end{itemize}

\section{Related Work}
\label{sec:related}

\paragraph{Sequence encoders.}
Sequential recommendation is dominated by two families. Self-attention encoders such as SASRec~\citep{kang2018sasrec}, BERT4Rec~\citep{sun2019bert4rec}, and frequency-augmented variants like FEARec~\citep{du2023fearec} parameterise pairwise interactions through dot-product attention with positional embeddings. State-Space Models~\citep{gu2023mamba}, ported to recommendation by Mamba4Rec~\citep{liu2024mamba4rec} and extended with bidirectional gating in SIGMA~\citep{liu2025sigma}, replace attention with selective state updates whose implicit decay is determined by a single set of state-transition parameters per layer. Both families compress the multi-scale temporal structure of user histories into a single mechanism inside the encoder.

\paragraph{Time-aware sequential recommenders.}
A separate line conditions the encoder on inter-event time. TiSASRec~\citep{li2020tisasrec} augments self-attention with a learnable embedding of relative time intervals, and MEANTIME~\citep{cho2020meantime} attaches a small bank of temporal embeddings to different attention heads. Both inject time \emph{inside} the encoder and share a single global parameterisation of the temporal kernel. Multi-rate temporal weighting has been explored in adjacent domains: Gradformer~\citep{liu2024gradformer} applies per-head exponential decay masks to graph transformers. \MARS imports the multi-rate principle into sequential recommendation while differing from prior time-aware methods on three axes: it operates as a post-encoder aggregator rather than an in-encoder modification; it parameterises $K$ user-conditioned exponential decays rather than a single shared schedule; and it remains agnostic to the underlying sequence encoder.

\paragraph{Aggregation strategies.}
Sequential recommenders also differ in how the per-position hidden states $\{\mathbf{h}_1,\dots,\mathbf{h}_L\}$ are reduced to a single user vector. The simplest choice is to read out the last position (SASRec~\citep{kang2018sasrec}, FEARec~\citep{du2023fearec}, Mamba4Rec~\citep{liu2024mamba4rec}). BERT4Rec~\citep{sun2019bert4rec} trains the encoder to recover masked tokens from bidirectional context, MEANTIME~\citep{cho2020meantime} averages multiple time-conditioned attention heads, and SIGMA~\citep{liu2025sigma} gates the bidirectional Mamba output before pooling. None of these aggregators consume real timestamps as a direct input; the temporal information they exploit is whatever signal the encoder has already propagated through its hidden states. \MARS departs from this convention by treating aggregation as a small parametric module that consumes timestamps explicitly and produces $K$ recency-conditioned summaries fused by a context-adaptive gate. Table~\ref{tab:related} contrasts \MARS with the closest prior work along five such design axes.

\begin{table}[t]
\centering
\caption{Sequential recommenders compared along five design axes for time-aware modelling.}
\label{tab:related}
\footnotesize
\setlength{\tabcolsep}{4pt}
\begin{tabular}{lccccc}
\toprule
\textbf{Method} & \textbf{Time} & \textbf{\#Decay} & \textbf{User-cond.} & \textbf{Fusion} & \textbf{Backbone} \\
\midrule
SASRec~\citep{kang2018sasrec}      & position       & $\times$        & $\times$    & last-token    & Trans. \\
BERT4Rec~\citep{sun2019bert4rec}   & position       & $\times$        & $\times$    & masked LM     & Trans. \\
FEARec~\citep{du2023fearec}        & frequency      & $\times$        & $\times$    & last-token    & Trans. \\
TiSASRec~\citep{li2020tisasrec}    & rel.\ interval & 1               & $\times$    & last-token    & Trans. \\
MEANTIME~\citep{cho2020meantime}   & multi-temp.    & per-head        & $\times$    & head avg.     & Trans. \\
Mamba4Rec~\citep{liu2024mamba4rec} & implicit       & 1/layer         & $\times$    & last-token    & Mamba \\
SIGMA~\citep{liu2025sigma}         & implicit       & bi-dir.         & $\times$    & gated         & Mamba \\
\midrule
\textbf{\MARS (Ours)}    & timestamp      & $K$ rates       & $\checkmark$ & adaptive gate & Trans./Mamba \\
\bottomrule
\end{tabular}
\end{table}

\section{Method: \MARS}
\label{sec:method}

\subsection{Problem Definition}
\label{sec:prelim}

Let $\mathcal{U}$ and $\mathcal{V}$ denote the user and item sets. A user $u\in\mathcal{U}$ has an ordered interaction history $S_u=\{(v_i, t_i)\}_{i=1}^{L_u}$, where $v_i\in\mathcal{V}$ is the $i$-th interacted item, $t_i$ is its UNIX timestamp, and $L_u$ is the user's history length. Histories are padded or truncated to a fixed maximum length $L$. The sequential recommendation task is to predict $v_{L_u+1}$, ranking items in $\mathcal{V}$ by predicted compatibility with the user's history. Throughout the paper we use $d$ for the hidden dimension, $K$ for the number of decay heads, $\mathbf{H}=[\mathbf{h}_1,\ldots,\mathbf{h}_L]\in\mathbb{R}^{L\times d}$ for the encoder output, $\Delta t_i = t_{L_u}-t_i \geq 0$ for the elapsed time between event $i$ and the latest event $L_u$, and $\mathbf{m}\in\{0,1\}^L$ for the binary padding mask that takes value $0$ at padded positions.

\subsection{Overview}
\label{sec:overview}

The architecture of \MARS is shown in Figure~\ref{fig:arch}. A sequence encoder transforms the item history into hidden states $\mathbf{H}\in\mathbb{R}^{L\times d}$, which the post-encoder \MARS module processes in three stages: per-user multi-rate recency aggregation that produces $K$ time-scale summaries (Section~\ref{sec:agg}), context-adaptive fusion that combines them into a single user vector (Section~\ref{sec:fuse}), and joint optimisation of all parameters with two auxiliary regularisers (Section~\ref{sec:opt}). Two encoder instantiations are studied under an identical protocol: \MARSv stacks $N$ Transformer blocks with causal masking on top of a learnable positional embedding, inheriting the configuration of SASRec~\citep{kang2018sasrec}; \MARSm replaces the Transformer blocks with $N$ Mamba blocks following Mamba4Rec~\citep{liu2024mamba4rec}, with a feed-forward sub-layer per block. We denote $\mathbf{h}_{\text{last}}\in\mathbb{R}^d$ the row of $\mathbf{H}$ at the user's last unmasked position $L_u$.

\begin{figure}[t]
\centering
\includegraphics[width=0.9\textwidth]{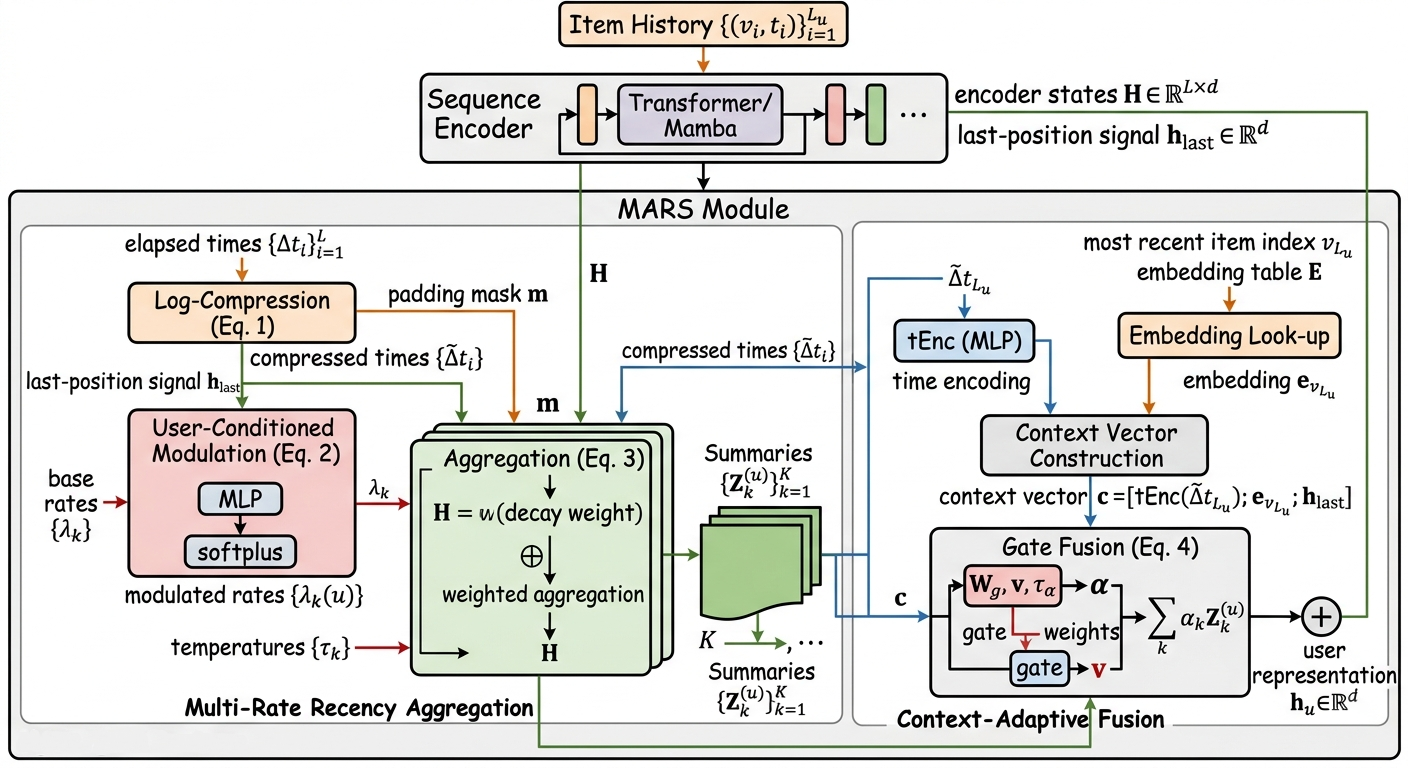}
\caption{\MARS architecture. A sequence encoder produces hidden states $\mathbf{H}$. \MARS aggregates them into $K$ recency-conditioned summaries via learnable rates $\{\lambda_k\}_{k=1}^K$ that are modulated per user, then fuses them into the user vector $\mathbf{h}_u$ through a context-adaptive gate.}
\label{fig:arch}
\end{figure}

\subsection{Multi-Rate Recency Aggregation}
\label{sec:agg}

\paragraph{Per-head decay weighting.}
The elapsed times $\{\Delta t_i\}_{i=1}^{L}$ are first log-compressed to
\begin{equation}
\tilde{\Delta}t_i \;=\; \log\!\bigl(1 + \Delta t_i / \tau_0\bigr),
\label{eq:logtime}
\end{equation}
with $\tau_0$ fixed at one hour. The compression keeps the dynamic range bounded across data regimes whose inter-event gaps may span several orders of magnitude; padding positions inherit $\tilde{\Delta}t_i=0$. Each of the $K$ decay heads carries two scalar parameters, a base rate $\lambda_k=\mathrm{softplus}(\tilde\lambda_k)$ and a temperature $\tau_k=\mathrm{softplus}(\tilde\tau_k)+0.1$. The unconstrained parameter $\tilde\lambda_k$ is initialised on a linspace from $-1$ to $+2$, so the heads begin from a diverse decomposition of the time axis.

\paragraph{User-conditioned modulation.}
To inject per-user heterogeneity without per-user parameters, the base rate of head $k$ is modulated by a small MLP applied to the user summary $\mathbf{h}_{\text{last}}$,
\begin{equation}
\lambda_k(u) \;=\; \lambda_k \cdot \exp\!\bigl(\sigma\cdot \tanh(\mathbf{W}_\lambda\,\mathbf{h}_{\text{last}})_k\bigr),
\label{eq:userlambda}
\end{equation}
with $\mathbf{W}_\lambda\in\mathbb{R}^{K\times d}$ and $\sigma=0.5$ bounding the modulation magnitude. Zero-initialising the final MLP layer ensures training starts from the unmodulated regime $\lambda_k(u)=\lambda_k$.

\paragraph{Aggregation.}
The $K$ summaries are obtained by softmax-normalising the per-head decay logits over the unmasked positions and aggregating $\mathbf{H}$ accordingly,
\begin{equation}
w_{k,i}^{(u)} \;=\; \frac{\exp(-\lambda_k(u)\,\tilde{\Delta}t_i / \tau_k)\, m_i}{\sum_{j=1}^{L}\exp(-\lambda_k(u)\,\tilde{\Delta}t_j / \tau_k)\, m_j},
\qquad
\mathbf{Z}_k^{(u)} \;=\; \sum_{i=1}^{L} w_{k,i}^{(u)}\,\mathbf{h}_i.
\label{eq:mtda}
\end{equation}
The matrix-vector formulation $\mathbf{Z} = \mathbf{w}\,\mathbf{H}\in\mathbb{R}^{K\times d}$ collapses the entire aggregation into a single batched matrix multiplication.
The choice of a $K$-rate exponential mixture as the aggregation kernel is motivated by the following approximation result.

\begin{proposition}[Approximation by exponential mixtures]
\label{prop:approx}
Let $\phi:[0,T]\to[0,1]$ be a monotonically non-increasing function with $\phi(0)=1$. For every $\varepsilon>0$ there exist $K\in\mathbb{N}$, weights $\{\alpha_k\}_{k=1}^K\subset\mathbb{R}$ with $\sum_k\alpha_k=1$, and rates $\{\lambda_k\}_{k=1}^K\subset\mathbb{R}_{\ge 0}$ such that
\begin{equation}
\sup_{t\in[0,T]}\Bigl|\,\phi(t)\;-\;\sum_{k=1}^{K}\alpha_k\exp(-\lambda_k\,t)\Bigr| \;\le\;\varepsilon.
\end{equation}
\end{proposition}

A complete proof is given in Appendix~\ref{app:proof}. Equation~\eqref{eq:mtda} realises such a mixture once the aggregation is combined with the fusion gate of Section~\ref{sec:fuse}, with the gate output playing the role of the mixture weights and the user-conditioned rates $\lambda_k(u)$ playing the role of $\lambda_k$. Proposition~\ref{prop:approx} therefore states that, in principle, \MARS can represent any monotone decay schedule the data demands; $K$ controls the fidelity of the approximation rather than the existence of a representable solution.

\subsection{Context-Adaptive Fusion}
\label{sec:fuse}

The $K$ summaries $\{\mathbf{Z}_k^{(u)}\}_{k=1}^K$ are fused through a gate conditioned on a context vector $\mathbf{c}=[\,\mathrm{tEnc}(\tilde{\Delta}t_{L_u});\;\mathbf{e}_{v_{L_u}};\;\mathbf{h}_{\text{last}}\,]$, where $\mathrm{tEnc}$ is a two-layer MLP applied to the most recent log-compressed elapsed time and $\mathbf{e}_{v_{L_u}}\in\mathbb{R}^d$ is the embedding of the most recent item:
\begin{equation}
g_k = \mathbf{v}^{\!\top}\!\tanh\!\bigl(\mathbf{W}_g[\mathbf{Z}_k^{(u)};\,\mathbf{c}]\bigr),
\quad
\boldsymbol{\alpha} = \mathrm{softmax}(\mathbf{g}/\tau_\alpha),
\quad
\mathbf{h}_u \;=\; \mathbf{h}_{\text{last}} + \sum_{k=1}^{K}\alpha_k\,\mathbf{Z}_k^{(u)},
\label{eq:gate}
\end{equation}
with learnable parameters $\mathbf{W}_g$, $\mathbf{v}$, and the gate temperature $\tau_\alpha>0$. The trailing residual connection preserves the encoder's last-position signal, ensuring the modified user vector cannot be worse than the encoder readout in the limit $\boldsymbol{\alpha}\to\mathbf{0}$. Algorithm~\ref{alg:mars} summarises the full forward pass.

The cost of the full pass scales linearly in sequence length and in the number of heads.

\begin{proposition}[Complexity]
\label{prop:complexity}
\MARS requires $\mathcal{O}(L\,d\,K)$ time and $\mathcal{O}(L\,K + d\,K)$ additional memory beyond the encoder's hidden states $\mathbf{H}$. The user-conditioned modulation contributes an additional $\mathcal{O}(d\,K)$ per user; the context-adaptive fusion contributes $\mathcal{O}(d\,K)$.
\end{proposition}
\begin{proof}
The aggregation in Eq.~\eqref{eq:mtda} is a single matrix product $\mathbf{w}\,\mathbf{H}$ between $\mathbf{w}\in\mathbb{R}^{K\times L}$ and $\mathbf{H}\in\mathbb{R}^{L\times d}$, costing $\mathcal{O}(LdK)$ multiplications. Computing the logits in Step~3 of Algorithm~\ref{alg:mars} costs $\mathcal{O}(LK)$, the user-modulation MLP in Eq.~\eqref{eq:userlambda} costs $\mathcal{O}(dK)$ per user, and the fusion gate in Eq.~\eqref{eq:gate} contributes another $\mathcal{O}(dK)$.
\end{proof}

For comparison, scaled dot-product self-attention is $\mathcal{O}(L^2 d)$ in time and $\mathcal{O}(L^2)$ in attention memory. Whenever $K\!\ll\!L$, the \MARS overhead is dominated by the encoder cost.

\subsection{Optimisation}
\label{sec:opt}

The standard sequential cross-entropy loss is augmented with two auxiliary regularisers,
\begin{equation}
\mathcal{L} \;=\; \mathcal{L}_{\text{CE}} \;+\; \eta_{\text{div}}\,\mathcal{L}_{\text{div}} \;+\; \eta_{\text{bal}}\,\mathcal{L}_{\text{bal}},
\label{eq:loss}
\end{equation}
with default coefficients $\eta_{\text{div}}=\eta_{\text{bal}}=10^{-2}$. The diversity term $\mathcal{L}_{\text{div}}=-\sum_{k\neq k'}\mathrm{JSD}(w_{k,\cdot}^{(u)}\Vert w_{k',\cdot}^{(u)})$ is the negative pairwise Jensen--Shannon divergence of the per-head attention distributions, and the load-balance term $\mathcal{L}_{\text{bal}} = K\sum_{k=1}^{K}\bar{\alpha}_k^{2}$, with $\bar{\alpha}_k$ the batch mean of the $k$-th gate weight, is the MoE-style regulariser that prevents the gate from routing through a single head.

The role of $\mathcal{L}_{\text{div}}$ is theoretically motivated. Without it, gradient descent can collapse all $K$ heads onto a single rate, in which case the multi-rate expressive power promised by Proposition~\ref{prop:approx} is lost. We rule this collapse out under a Hawkes-process generative model.

\begin{proposition}[Identifiability of decay rates]
\label{prop:identif}
Suppose user behaviour is generated from a multivariate Hawkes process with $K$ exponential excitation kernels of distinct positive rates $\lambda_1^*<\dots<\lambda_K^*$ and non-degenerate mixture weights $\alpha_k^*>0$, $\sum_k\alpha_k^*=1$. If $\mathcal{L}_{\mathrm{div}}$ is strictly positive at the optimum, then the learned rates $\{\lambda_k\}_{k=1}^K$ recover $\{\lambda_k^*\}_{k=1}^K$ up to permutation in the population limit, and the learned weights $\{\alpha_k\}$ recover $\{\alpha_k^*\}$ in $\ell^1$ norm.
\end{proposition}
\begin{proof}[Sketch]
The Hawkes intensity~\citep{hawkes1971spectra} is a $K$-mixture of exponentials, so by classical identifiability of finite mixtures of distinct exponential laws~\citep{teicher1963identifiability}, the mixing measure is identifiable. The diversity loss strictly penalises any two heads collapsing onto the same exponential ($\lambda_k=\lambda_{k'}\Rightarrow w_{k,\cdot}=w_{k',\cdot}\Rightarrow\mathrm{JSD}=0$), so optimisation cannot exit the strict-positivity region; this enforces distinctness of the rates. Combined with the load-balancing constraint $\sum_k\alpha_k=1$, $\alpha_k>0$, the population minimiser of $\mathcal{L}$ is unique up to head permutation. The detailed proof is in Appendix~\ref{app:identif}.
\end{proof}

\subsection{Backbone Selection}
\label{sec:select}

Given a target dataset, \MARS\ selects between \MARSv\ and \MARSm\ based on a single dataset-level statistic, the average training-sequence length $\bar{L}$:
\begin{equation}
\text{backbone}(\mathcal{D}) = 
\begin{cases}
\MARSv & \text{if } \bar{L}(\mathcal{D}) < \tau, \\
\MARSm & \text{if } \bar{L}(\mathcal{D}) \geq \tau,
\end{cases}
\qquad \tau = 50.
\label{eq:select}
\end{equation}
The threshold $\tau$ reflects the empirical regime boundary observed in our analysis: SASRec attention entropy is $0.75$--$0.87$ (near-uniform) on the four sparse benchmarks where $\bar{L}\in[8,13]$, and an explicit recency prior delivers $+4$ to $+19\%$ HR@10; on dense ML-1M where $\bar{L}=165$, entropy collapses to $0.43$--$0.51$ and Mamba's selective state encodes the recency prior more effectively than a stacked aggregator (Appendix~\ref{app:entropy}). The $20\!\times$ gap between $\bar{L}$ on sparse and dense benchmarks makes the rule robust to any threshold in $(13, 165)$.

The selection is computed once before training and adds no runtime overhead. We deliberately keep it as a dataset-statistic heuristic rather than a learned router: data-density regime is a property of the deployment dataset, not of individual users, so a learned mechanism would degenerate into the same threshold while incurring additional training cost.

\section{Experiments}
\label{sec:exp}

\subsection{Settings}

\textbf{Datasets.} We evaluate on five widely-used benchmarks summarised in Table~\ref{tab:datasets}, namely the three Amazon categories Beauty, Sports, and Games released by \citet{he2016amazon}, MovieLens-1M~\citep{harper2015movielens}, and Yelp~\citep{yelp_dataset}. Beauty, Sports, Games and Yelp belong to the sparse, short-history regime, with average sequence length below $13$; ML-1M is a dense regime with sequences nearly twenty times longer. This split lets us probe \MARS across the full data-density spectrum encountered in standard sequential-recommendation work.

\begin{wraptable}{r}{0.44\textwidth}
    \centering
    \footnotesize
    \vspace{-15pt}
    \caption{Statistics of the benchmarks.}
    \label{tab:datasets}
    \vspace{-5pt}
    \setlength{\tabcolsep}{1pt}
    \begin{tabular}{lrrrrr}
    \toprule
    \textbf{Dataset} & \textbf{\#Users} & \textbf{\#Items} & \textbf{\#Inters.} & \textbf{Sparsity} & \textbf{Avg.} \\
    \midrule
    Beauty & 22,363 & 12,101 & 198,502 & 99.93\% & 8.88 \\
    Sports & 35,598 & 18,357 & 296,337 & 99.95\% & 8.32 \\
    Games  & 94,762 & 25,612 & 814,586 & 99.97\% & 8.60 \\
    ML-1M  & 6,040  & 3,416  & 999,611 & 95.16\% & 165.5 \\
    Yelp   & 233,247 & 109,062 & 2,918,699 & 99.99\% & 12.51 \\
    \bottomrule
    \end{tabular}
    \vspace{-10pt}
\end{wraptable}

\textbf{Baselines.} Ten representative baselines spanning RNN, attention-based, time-aware, session-based, and state-space families are re-run under our unified RecBole~\citep{zhao2021recbole} framework. They are the recurrent GRU4Rec~\citep{hidasi2016gru4rec}, the attention-augmented recurrent NARM~\citep{li2017narm}, the bidirectional Transformer BERT4Rec~\citep{sun2019bert4rec}, the causal Transformer SASRec~\citep{kang2018sasrec}, the time-aware Transformer TiSASRec~\citep{li2020tisasrec}, the frequency-domain FEARec~\citep{du2023fearec}, the consistent-representation session model CORE~\citep{hou2022core}, and three state-space models, namely Mamba4Rec~\citep{liu2024mamba4rec}, EchoMamba4Rec (denoted Echo in tables)~\citep{wang2024echomamba4rec}, and SIGMA~\citep{liu2025sigma}. Every baseline is re-trained on the same splits to avoid pipeline confounds, and numbers from the original publications are not quoted.

\textbf{Implementation details.} All Transformer-based models share the SASRec configuration of $d{=}64$, two attention blocks, inner dimension $256$, and dropout $0.5$ on Beauty, Sports, Yelp and $0.2$ on Games, ML-1M. Mamba-based models use the Mamba4Rec defaults of $d_{\text{state}}{=}32$, $d_{\text{conv}}{=}4$, $\text{expand}{=}2$. Optimisation uses Adam~\citep{kingma2015adam} at learning rate $10^{-3}$ with batch size $2048$ for up to $200$ epochs and early stopping on validation NDCG@10 with patience $20$. \MARS uses $K{=}4$ decay heads on the four sparse benchmarks and $K{=}8$ on ML-1M to match the longer sequences. All experiments are performed on a single NVIDIA V100 GPU. Each \MARS configuration is re-trained with five random seeds and we report mean$\pm$std; the baselines use a single seed each.

\subsection{Comparison}
\label{sec:main}

\begin{table}[t]
\centering
\caption{Comparison results ($\times 100$). \MARS\ reports mean$\pm$std over five seeds. Backbone is selected by $\bar{L}$, with $^\dagger$ marking \MARSv\ ($\bar{L}<50$) and $^\ddagger$ marking \MARSm\ ($\bar{L}\geq 50$). \textbf{Best}, \underline{second}.}
\label{tab:main}
\setlength{\tabcolsep}{1pt} \resizebox{\textwidth}{!}{%
\begin{tabular}{ll|cccccccccc|ccc}
\toprule
\textbf{Dataset} & \textbf{Metric} & \textbf{GRU4Rec} & \textbf{BERT4Rec} & \textbf{NARM} & \textbf{SASRec} & \textbf{TiSASRec} & \textbf{FEARec} & \textbf{CORE} & \textbf{Mamba4Rec} & \textbf{Echo} & \textbf{SIGMA} & \textbf{\MARSv} & \textbf{\MARSm} & \textbf{\MARS} \\
\midrule
\multirow{3}{*}{Beauty}
 & HR@10 & 5.82 & 3.94 & 6.03 & 8.04 & \second{8.78} & 7.83 & 6.49 & 6.62 & 6.93 & 6.03 & 8.84\tiny{$\pm$0.14} & 8.02\tiny{$\pm$0.23} & \best{8.84$^\dagger$} \\
 & NDCG  & 3.25 & 1.95 & 3.45 & 4.10 & \best{4.45} & 4.08 & 2.59 & 4.12 & 4.30 & 3.43 & \second{4.40\tiny{$\pm$0.06}} & 4.13\tiny{$\pm$0.08} & 4.40$^\dagger$ \\
 & MRR   & 2.47 & 1.36 & 2.67 & 2.90 & 3.12 & 2.92 & 1.44 & \second{3.36} & \best{3.50} & 2.64 & 3.04\tiny{$\pm$0.06} & 2.94\tiny{$\pm$0.04} & 3.04$^\dagger$ \\
\midrule
\multirow{3}{*}{Sports}
 & HR@10 & 3.06 & 1.66 & 3.12 & 4.53 & \second{4.84} & 4.60 & 3.24 & 3.50 & 3.40 & 3.20 & 5.22\tiny{$\pm$0.08} & 4.73\tiny{$\pm$0.19} & \best{5.22$^\dagger$} \\
 & NDCG  & 1.58 & 0.81 & 1.60 & 2.14 & \second{2.30} & 2.19 & 1.26 & 2.08 & 2.08 & 1.77 & 2.43\tiny{$\pm$0.04} & 2.26\tiny{$\pm$0.07} & \best{2.43$^\dagger$} \\
 & MRR   & 1.13 & 0.56 & 1.14 & 1.41 & 1.52 & 1.44 & 0.69 & \second{1.65} & \best{1.67} & 1.33 & 1.57\tiny{$\pm$0.03} & 1.51\tiny{$\pm$0.05} & 1.57$^\dagger$ \\
\midrule
\multirow{3}{*}{Games}
 & HR@10 & 7.80 & 6.69 & 10.15 & 8.86 & \second{12.07} & 8.64 & 8.84 & 8.02 & 8.02 & 8.10 & 12.07\tiny{$\pm$0.16} & 12.01\tiny{$\pm$0.42} & \best{12.07$^\dagger$} \\
 & NDCG  & 4.05 & 3.49 & 5.24 & 4.34 & \best{5.90} & 4.20 & 3.61 & 4.21 & 4.18 & 4.20 & 5.51\tiny{$\pm$0.09} & \second{5.62\tiny{$\pm$0.11}} & 5.51$^\dagger$ \\
 & MRR   & 2.92 & 2.53 & \second{3.77} & 2.97 & \best{4.03} & 2.84 & 2.06 & 3.05 & 3.01 & 3.03 & 3.52\tiny{$\pm$0.07} & 3.68\tiny{$\pm$0.07} & 3.52$^\dagger$ \\
\midrule
\multirow{3}{*}{ML-1M}
 & HR@10 & 29.37 & 23.11 & 29.37 & 29.62 & 24.90 & 28.05 & 15.55 & 30.28 & 31.04 & \second{31.79} & 29.04\tiny{$\pm$0.70} & 32.80\tiny{$\pm$0.45} & \best{32.80$^\ddagger$} \\
 & NDCG  & 16.90 & 11.87 & 16.58 & 16.61 & 13.42 & 15.63 & 6.79 & 17.72 & 18.00 & \second{18.63} & 16.39\tiny{$\pm$0.50} & 18.80\tiny{$\pm$0.32} & \best{18.80$^\ddagger$} \\
 & MRR   & 13.10 & 8.48 & 12.67 & 12.65 & 9.93 & 11.86 & 4.18 & 13.88 & 14.03 & \best{14.60} & 12.52\tiny{$\pm$0.45} & \second{14.52\tiny{$\pm$0.29}} & 14.52$^\ddagger$ \\
\midrule
\multirow{3}{*}{Yelp}
 & HR@10 & 2.98 & 2.38 & 3.42 & 5.41 & 5.62 & 4.79 & \second{6.14} & 4.35 & 3.79 & 3.80 & 6.45\tiny{$\pm$0.04} & 5.97\tiny{$\pm$0.04} & \best{6.45$^\dagger$} \\
 & NDCG  & 1.65 & 1.21 & 1.89 & 3.64 & 3.84 & 3.42 & 3.65 & 2.55 & 2.13 & 2.12 & 4.22\tiny{$\pm$0.01} & \second{3.89\tiny{$\pm$0.01}} & \best{4.22$^\dagger$} \\
 & MRR   & 1.25 & 0.86 & 1.42 & 3.10 & \second{3.30} & 3.00 & 2.89 & 2.00 & 1.63 & 1.61 & 3.53\tiny{$\pm$0.01} & 3.25\tiny{$\pm$0.02} & \best{3.53$^\dagger$} \\
\bottomrule
\end{tabular}}
\end{table}

The two \MARS\ instantiations are compared in Table~\ref{tab:main} against ten re-trained baselines spanning RNN, Transformer, time-aware Transformer, and Mamba families. We highlight three observations.
\textbf{(i)~\MARSv\ wins HR@10 on every sparse benchmark.} \MARSv\ attains the strongest HR@10 across Beauty, Sports, Games, and Yelp with mean gain $+19.7\%$ over the strongest content-only Transformer baseline ($+9.95\%$ Beauty, $+13.5\%$ Sports, $+36.2\%$ Games, $+19.2\%$ Yelp) and $+5.85\%$ over TiSASRec at $\mathcal{O}(LdK)$ cost versus TiSASRec's $\mathcal{O}(L^2)$. \MARSv\ also leads NDCG@10 on Sports and Yelp; TiSASRec marginally outperforms by $1.1\%$ and $7.1\%$ on Beauty and Games NDCG, consistent with its per-pair interval matrix favouring fine-grained ranking. MRR is mixed, with Echo leading on Beauty and Sports and TiSASRec on Games, reflecting last-token readouts being well-calibrated for short-tail interactions while \MARSv\ aggregates over the full backbone for HR-style top-$k$ retrieval.
\textbf{(ii)~\MARSm\ leads HR@10 and NDCG on the dense ML-1M.} On dense ML-1M ($\bar{L}=165$), \MARSm\ overtakes every prior method on HR@10 and NDCG, beating SIGMA by $+3.2\%$ and $+0.9\%$ and Mamba4Rec by $+8.3\%$ HR@10. On MRR \MARSm\ trails SIGMA by $0.55\%$ ($14.52$ vs $14.60$), within seed variance. The HR@10 lead is delivered at $42\%$ fewer MFLOPs than SIGMA (Section~\ref{sec:eff}), placing \MARSm\ on the accuracy-efficiency Pareto frontier. On sparse benchmarks \MARSm\ tracks \MARSv\ within one standard deviation on Games and Yelp but does not lead, motivating the dual-instantiation design.
\textbf{(iii)~Stable and reproducible.} Five-seed standard deviations lie between $0.04$ and $0.70$ HR@10 points across all $5\times 2$ cells, confirming gains are not seed artefacts.

\subsection{Ablation Study}
\label{sec:abl}

\begin{wraptable}{r}{0.50\textwidth}
\vspace{-13pt}
\centering
\footnotesize
\caption{Ablation of \MARSv (HR@10, $\times 100$).}
\label{tab:abl}
\vspace{-5pt}
\footnotesize \setlength{\tabcolsep}{1pt}
\resizebox{\linewidth}{!}{
\begin{tabular}{lccccc}
\toprule
\textbf{Variant} & \textbf{Beauty} & \textbf{Sports} & \textbf{Games} & \textbf{ML-1M} & \textbf{Yelp} \\
\midrule
Full \MARSv                  & \best{8.91} & \best{5.20} & \best{12.24} & 29.30 & \best{6.40} \\
Backbone only (w/o \MARS)    & 8.53 & 4.89 & 11.75 & \best{30.33} & 5.36 \\
b1 (w/o real $\Delta t$)     & 8.69 & 5.16 & 12.32 & 29.82 & 6.49 \\
b2 (w/o $\lambda(u)$)        & 8.94 & 5.21 & 12.32 & 29.25 & 6.39 \\
b3 (w/o diversity)           & 8.90 & 5.21 & 12.14 & 29.22 & 6.37 \\
b4 (w/o balance)             & 8.88 & 5.12 & 12.26 & 28.41 & 6.33 \\
Single head ($K{=}1$)        & 8.80 & 5.27 & 12.23 & 27.45 & 6.42 \\
\bottomrule
\end{tabular}
}
\vspace{-10pt}
\end{wraptable}

\textbf{(i)~\MARS contributes between +4\% and +19\% on sparse data.} Comparing Full \MARSv against the Backbone-only variant isolates the marginal value of \MARS on the Transformer backbone. The gain is largest on Yelp at $+19.4\%$, where the bursty review pattern carries strong multi-rate structure, and stays positive on the three Amazon benchmarks at $+4.0\%$ to $+6.3\%$. These results validate \MARS as a lightweight enhancement of SASRec-style backbones for sparse, short-history data.

\textbf{(ii)~On dense ML-1M the Transformer backbone alone is competitive with Full \MARSv.} Backbone-only \MARSv attains $30.33$ HR@10, $1.03$ above the full module's $29.30$. SASRec's positional attention already concentrates on the recent tail of dense histories (normalised entropy $0.43$--$0.51$, Appendix~\ref{app:entropy}), so the multi-rate prior of \MARS adds limited information when stacked on top. This finding is the empirical motivation for our Mamba instantiation: Appendix~\ref{app:abl_full_marsm} shows that Full \MARSm exceeds the corresponding backbone-only readout on ML-1M across all three metrics, confirming that \MARS contributes positively when paired with an encoder whose intrinsic recency mechanism is more constrained than self-attention.

\textbf{(iii)~Internal components are mostly dataset-dependent.} The four component ablations (\texttt{b1}--\texttt{b4}) generally lie within $\pm 3\%$ of the Full model. Real timestamps (\texttt{b1}) carry the largest single contribution on Beauty, costing $-2.5\%$ when removed; the load-balance regulariser (\texttt{b4}) is most important on ML-1M, costing $-3.0\%$ when removed. The user-conditioned $\lambda$ (\texttt{b2}) and JSD diversity (\texttt{b3}) contribute smaller but positive amounts, and a single decay head ($K{=}1$) is competitive on the sparse benchmarks. Section~\ref{sec:hyper} expands this observation into a full sweep over $K$.

\subsection{Hyperparameter Sensitivity}
\label{sec:hyper}

\begin{figure}[t]
\centering
\includegraphics[width=0.99\textwidth]{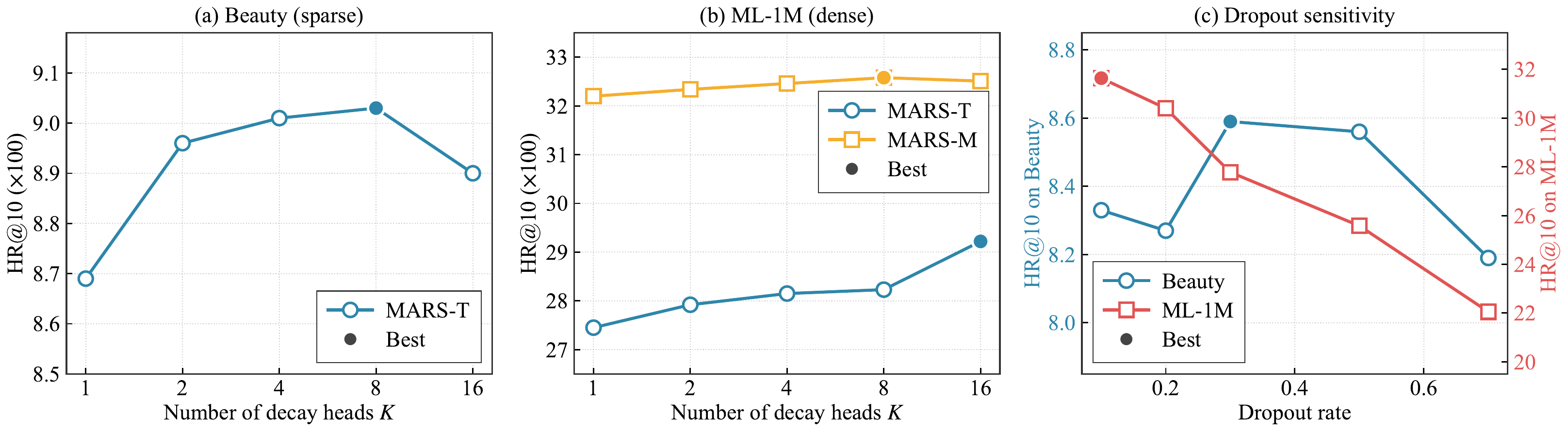}
\caption{Hyperparameter sensitivity of \MARS. (a) $K$ on Beauty, (b) $K$ on ML-1M for \MARSv\ and \MARSm, and (c) dropout on Beauty and ML-1M.}
\label{fig:hyper}
\vspace{-10pt}
\end{figure}

\textbf{(i)~$K$ is mildly influential on sparse data.} On the sparse benchmarks $K{=}1$ already attains $96.2\%$ of the best $K{=}8$ result on Beauty, consistent with Proposition~\ref{prop:approx} which guarantees only the existence of a $K$-mixture approximation; the empirical decay kernels of short-history data are smooth enough that a small mixture suffices.
\textbf{(ii)~$K$ matters substantially on dense ML-1M, more for \MARSv\ than for \MARSm.} For the Transformer-backboned \MARSv, HR@10 climbs from $27.45$ at $K{=}1$ to a peak of $29.22$ at $K{=}16$ ($+6.4\%$). For the Mamba-backboned \MARSm the spread is milder, $32.20\!\to\!32.58$ at $K{=}8$ ($+1.2\%$), reflecting that Mamba's selective state already encodes a strong recency prior so adding multiple decay rates yields diminishing returns. The contrast between the two backbones is itself a prediction of the multi-rate hypothesis: $K{>}1$ helps most when the encoder lacks an intrinsic mechanism for multi-scale temporal weighting. We therefore default to $K{=}4$ on the four sparse benchmarks and $K{=}8$ on ML-1M as a Pareto-balanced choice. The subgroup analysis further corroborates this by attributing the marginal value of multiple heads to the long-history user buckets.
\textbf{(iii)~Dropout is dataset-density-dependent.} Following established SASRec preferences, the optimum is $0.3$ on Beauty and $0.1$ on ML-1M; larger dropout rapidly degrades ML-1M, where moving from $0.1$ to $0.5$ costs over $6$ HR@10 points. Across the entire grid we explored, \MARS\ never under-performs the strongest baseline reported in Table~\ref{tab:main}.

\subsection{Efficiency Analysis}
\label{sec:eff}

\begin{wraptable}{r}{0.60\textwidth}
\vspace{-13pt}
\centering
\caption{Efficiency on Beauty and ML-1M (NVIDIA V100).}
\label{tab:efficiency}
\vspace{-5pt}
\footnotesize \setlength{\tabcolsep}{1pt}
\resizebox{\linewidth}{!}{
\begin{tabular}{c|cccc|cccc}
\toprule
\multirow{2}{*}{\textbf{Model}} & \multicolumn{4}{c|}{\textbf{Beauty}} & \multicolumn{4}{c}{\textbf{ML-1M}} \\
\cmidrule(lr){2-5} \cmidrule(lr){6-9}
 & \#P (K) & MFLOPs & Time (s) & Mem (G) & \#P (K) & MFLOPs & Time (s) & Mem (G) \\
\midrule
GRU4Rec   & 857 & 4.14 & \best{0.044} & 1.80 & 301 & 16.56 & \best{0.035} & 7.07 \\
BERT4Rec  & 894 & 5.32 & 0.907 & 1.19 & 339 & 22.70 & 0.503 & 11.97 \\
SASRec    & 878 & 5.10 & 0.197 & 1.22 & 332 & 21.83 & 0.313 & 12.56 \\
FEARec    & 878 & 12.43 & 0.726 & 1.40 & 332 & 49.72 & 0.726 & 5.93 \\
Mamba4Rec & \best{847} & \best{1.68} & 0.124 & \best{0.89} & \best{291} & \best{6.71} & \second{0.097} & \best{3.09} \\
Echo      & 864 & 2.52 & 0.348 & 1.38 & 308 & 10.09 & 0.276 & 5.05 \\
SIGMA     & 905 & 5.80 & 0.300 & 1.31 & 349 & 23.18 & 0.251 & 5.12 \\
\midrule
\best{\MARSv} & 896 & 5.05 & \second{0.151} & 1.05 & 350 & 20.27 & 0.173 & 4.56 \\
\best{\MARSm} & 937 & \second{3.40} & 0.237 & \second{0.89} & 382 & \second{13.48} & 0.181 & \second{3.11} \\
\bottomrule
\end{tabular}
}
\vspace{-10pt}
\end{wraptable}

Per-sample efficiency on Beauty and ML-1M is reported in Table~\ref{tab:efficiency}. \MARSv adds only $2\%$ parameters over SASRec on Beauty ($896$K vs $878$K) and $5\%$ on ML-1M ($350$K vs $332$K). Its forward MFLOPs are within $1\%$ of SASRec on Beauty ($5.05$ vs $5.10$) and $7\%$ smaller on ML-1M ($20.27$ vs $21.83$), since the \MARS forward pass is dominated by a single batched matrix multiplication rather than quadratic attention. \MARSm inherits the linear-time Mamba backbone, with MFLOPs $40\%$ below SIGMA on Beauty ($3.40$ vs $5.80$) and $42\%$ on ML-1M ($13.48$ vs $23.18$). Its $+11\%$ Beauty and $+31\%$ ML-1M parameter overhead over Mamba4Rec comes from the deeper two-block stack. These compute savings translate to wall-clock latency. On Beauty \MARSv\ is $2.0\times$ faster than SIGMA and $6.0\times$ faster than BERT4Rec. On ML-1M \MARSm\ is $1.4\times$ faster than SIGMA and $2.8\times$ faster than BERT4Rec. Both variants occupy the accuracy-efficiency Pareto frontier.

\subsection{Subgroup Analysis: Where Does \MARS\ Help Most?}
\label{sec:subgroup}

The comparison is stratified by training-sequence length $L_u$ in Table~\ref{tab:subgroup}, with bucket boundaries detailed in Appendix~\ref{app:subgroup}.
\textbf{(i)~Sparse benchmarks.} On Yelp the relative improvement grows monotonically from $+21.8\%$ on short to $+29.7\%$ on very long users, consistent with the multi-rate hypothesis that longer histories expose more time-scales than single-decay backbones can represent. Sports shows a $+149.8\%$ very-long gain, tempered by $n{=}89$. Games is essentially flat, indicating its encoder already captures the relevant signal for short users. The Beauty very-long bucket is the only sparse case where SASRec narrowly beats \MARSv, but with $n{=}153$ the gap is within seed variance ($\pm 0.7$ HR@10 from Table~\ref{tab:main}).
\textbf{(ii)~Dense ML-1M.} \MARSv\ trails SASRec by $-2.9\%$ on the very-long bucket, which is the dominant $76\%$ of test users where Transformer attention already concentrates on the recent tail. \MARSm\ consistently outperforms Mamba4Rec across all four buckets by $+5.7\%$ to $+9.0\%$, with the gap widening from short to long histories.
\textbf{(iii)~Synthesis.} The marginal value of \MARS\ scales with the temporal heterogeneity of user histories, modulated by the encoder's intrinsic recency capacity, motivating the dual-instantiation design.

\begin{table}[h]
\centering
\caption{Subgroup HR@10 ($\times 100$) by user history length $L_u$. Each cell: baseline\,$\to$\,\MARS ($\Delta\%$). The baseline is SASRec for the four sparse benchmarks and the ML-1M, \MARSv\ row, and Mamba4Rec for the ML-1M, \MARSm\ row. \textcolor{red}{Red} marks $\Delta<0$. Bucket sizes in Appendix~\ref{app:subgroup}.}
\label{tab:subgroup}
\footnotesize
\setlength{\tabcolsep}{3pt}
\resizebox{\textwidth}{!}{%
\begin{tabular}{c|cccc}
\toprule
\textbf{Dataset} & \textbf{Short} ($L_u \le 5$) & \textbf{Medium} (6--15) & \textbf{Long} (16--49) & \textbf{Very long} ($\ge$ 50) \\
\midrule
Beauty           & 6.57\,$\to$\,7.98 (+21.4)   & 7.66\,$\to$\,8.40 (+9.7)    & 12.89\,$\to$\,14.49 (+12.4) & 20.26\,$\to$\,18.95 (\textcolor{red}{-6.5}) \\
Sports           & 4.13\,$\to$\,5.30 (+28.5)   & 3.91\,$\to$\,4.97 (+27.2)   & 3.25\,$\to$\,4.35 (+33.7)   & 2.25\,$\to$\,5.62 (+149.8) \\
Games            & 12.05\,$\to$\,11.31 (\textcolor{red}{-6.1}) & 10.91\,$\to$\,10.94 (+0.3) & 10.09\,$\to$\,10.25 (+1.6) & 8.10\,$\to$\,8.91 (+10.0) \\
Yelp             & 5.69\,$\to$\,6.93 (+21.8)   & 5.20\,$\to$\,6.50 (+25.0)   & 4.36\,$\to$\,5.49 (+25.9)   & 3.64\,$\to$\,4.72 (+29.7) \\
\midrule
ML-1M (\MARSv)   & 7.92\,$\to$\,8.04 (+1.5)    & 18.43\,$\to$\,19.13 (+3.8)  & 27.51\,$\to$\,27.85 (+1.2)  & 30.85\,$\to$\,29.96 (\textcolor{red}{-2.9}) \\
ML-1M (\MARSm)   & 8.21\,$\to$\,8.68 (+5.7)    & 18.86\,$\to$\,20.32 (+7.7)  & 28.42\,$\to$\,30.97 (+9.0)  & 31.45\,$\to$\,34.01 (+8.1) \\
\bottomrule
\end{tabular}}
\end{table}

\section{Disucussion}
\label{sec:discussion}

\paragraph{A unified framework.} \MARS\ is encoder-agnostic by design: the same multi-rate operator gives two complementary instantiations, \MARSv\ for sparse data and \MARSm\ for dense data, automatically selected by $\bar{L}$ before training (Section \ref{sec:select}). \MARSv\ adds a recency prior on top of near-uniform SASRec attention (sparse), while \MARSm\ relaxes Mamba's single-rate selective state through multiple learnable rates (dense), so the two variants cover the data-density spectrum without retraining a single architecture.

\paragraph{Mechanism.} The entropy gap (sparse 0.75--0.87 vs dense 0.43--0.51; Appendix~\ref{app:entropy}) explains where \MARS\ helps: an explicit recency prior delivers +4 to +19\% HR@10 on sparse benchmarks but is redundant on dense ML-1M, where we switch to Mamba. The $K$-sensitivity gap (96\% of the best at $K{=}1$ on sparse vs +6.4\% swing on dense) confirms Proposition~\ref{prop:identif}.

\paragraph{Limitations.} The backbone-selection rule (Eq.~\ref{eq:select}) is a hand-crafted dataset-level threshold rather than a learned mechanism, robust within the $20\!\times$ density gap of our benchmarks but unproven on richer regime spectra. The identifiability of Proposition~\ref{prop:identif} relies on a Hawkes-process model, with a distribution-free version remaining open. The $\mathcal{O}(L^2)$ TiSASRec marginally outperforms \MARSv\ on NDCG/MRR for Beauty and Games by $1.1$--$7.1\%$, a trade-off we accept given \MARS's $40\!\times$ memory advantage at long $L$. We evaluate only offline next-item prediction on five public benchmarks.

\paragraph{Broader impact.} Sequential recommenders shape consumption and exposure diversity on commercial platforms. \MARS\ inherits the standard feedback-loop risks of collaborative filtering without introducing new ones, with its sensitivity to recent activity helping recover from intent shifts but also amplifying transient exploration. We use only public benchmark datasets and propose no technique aimed at surveillance or behavioural manipulation.

\section{Conclusion}

We introduced \MARS, a time-aware aggregation operator that learns $K$ user-conditioned exponential decay rates and fuses them through a context-adaptive gate. \MARS\ automatically selects between \MARSv\ (Transformer encoder) for sparse data and \MARSm\ (Mamba encoder) for dense data based on a single dataset-level statistic, the average sequence length $\bar{L}$. \MARS\ attains the best HR@10 on every benchmark, with mean gain +19.7\% over the strongest content-only Transformer baseline on sparse data at $\mathcal{O}(LdK)$ cost, and +3.2\% HR@10 / +0.9\% NDCG over SIGMA on dense ML-1M at 42\% fewer MFLOPs, on the accuracy-efficiency Pareto frontier. A backbone-only ablation validates the regime-dependent design.

\clearpage
\newpage
\bibliographystyle{plainnat}
\bibliography{main}

\newpage
\clearpage
\appendix
\section*{Appendix}
\addcontentsline{toc}{section}{Appendix}

\renewcommand{\thetable}{A\arabic{table}}
\renewcommand{\thefigure}{A\arabic{figure}}
\renewcommand{\thealgorithm}{A\arabic{algorithm}}
\setcounter{table}{0}
\setcounter{figure}{0}
\setcounter{algorithm}{0}

\section{\MARS\ Forward Pass}
\label{app:algorithm}

The full \MARS\ forward pass for a single user is given in Algorithm~\ref{alg:mars}. Batching across users is performed via a single matrix multiplication on Step~5, with all other operations being pointwise or small per-user MLPs.

\begin{algorithm}[H]
\caption{\MARS\ forward pass.}
\label{alg:mars}
\begin{algorithmic}[1]
\Require Encoder outputs $\mathbf{H}\in\mathbb{R}^{L\times d}$, $\mathbf{h}_{\text{last}}\in\mathbb{R}^d$, mask $\mathbf{m}\in\{0,1\}^L$, elapsed times $\Delta t\in\mathbb{R}_{\ge 0}^{L}$, last item index $v_{L_u}$, embedding table $\mathbf{E}$;
         learnable rates $\boldsymbol{\lambda},\boldsymbol{\tau}\in\mathbb{R}_{>0}^{K}$, modulation $\mathbf{W}_\lambda$, fusion $(\mathbf{W}_g,\mathbf{v})$, time encoder $\mathrm{tEnc}$;
         hyperparameters $\sigma, \tau_\alpha, \tau_0$.
\Ensure user representation $\mathbf{h}_u\in\mathbb{R}^d$.
\State $\tilde{\Delta} t_i \gets \log\!\bigl(1 + \Delta t_i / \tau_0\bigr)\quad\forall i$
       \Comment{compress gaps}
\State $\boldsymbol{\delta} \gets \tanh(\mathbf{W}_\lambda\,\mathbf{h}_{\text{last}})$,\quad
       $\boldsymbol{\lambda}^{(u)} \gets \boldsymbol{\lambda}\odot\exp(\sigma\,\boldsymbol{\delta})$
       \Comment{user-conditioned rates}
\State $\ell_{k,i} \gets -\lambda_k^{(u)}\,\tilde{\Delta}t_i / \tau_k\quad\forall k,i$
\State $\ell_{k,i} \gets \ell_{k,i} - \infty\cdot(1-m_i)$,\quad
       $w_{k,i}\gets\mathrm{softmax}_i(\ell_{k,i})$
\State $\mathbf{Z}\gets \mathbf{w}\,\mathbf{H}\in\mathbb{R}^{K\times d}$
       \Comment{$\mathcal{O}(LdK)$}
\State $\mathbf{c}\gets[\,\mathrm{tEnc}(\tilde{\Delta}t_{L_u});\;\mathbf{e}_{v_{L_u}};\;\mathbf{h}_{\text{last}}\,]$
\State $g_k\gets\mathbf{v}^\top\tanh\!\bigl(\mathbf{W}_g[\mathbf{Z}_k;\mathbf{c}]\bigr)$,\quad
       $\boldsymbol{\alpha}\gets\mathrm{softmax}(\mathbf{g}/\tau_\alpha)$
\State $\mathbf{h}_u\gets \mathbf{h}_{\text{last}} + \sum_{k=1}^{K}\alpha_k\,\mathbf{Z}_k$
       \Comment{fuse + residual}
\State \Return $\mathbf{h}_u$
\end{algorithmic}
\end{algorithm}

\section{Proof of Proposition~\ref{prop:approx}}
\label{app:proof}

\noindent\textbf{Proposition~1.}\; \emph{Let $\phi:[0,T]\to[0,1]$ be a monotonically non-increasing function with $\phi(0)=1$. For every $\varepsilon>0$ there exist $K\in\mathbb{N}$, weights $\{\alpha_k\}_{k=1}^K\subset[0,1]$ with $\sum_k\alpha_k=1$, and rates $\{\lambda_k\}_{k=1}^K\subset\mathbb{R}_{\ge 0}$ such that $\sup_{t\in[0,T]}\bigl|\phi(t)-\sum_{k}\alpha_k e^{-\lambda_k t}\bigr| \le\varepsilon$.} \medskip

\begin{proof}
We proceed in three steps.

\emph{Step 1 (algebra of exponentials).} Let $\mathcal{E}=\{e^{-\lambda t}\,:\,\lambda\ge 0\}\subset C([0,T])$ and let $\mathcal{A}$ be the unital algebra generated by $\mathcal{E}$ over $\mathbb{R}$. Since $e^{-\lambda_1 t}\cdot e^{-\lambda_2 t}= e^{-(\lambda_1+\lambda_2)t}\in\mathcal{E}$ and constants $1=e^{-0\cdot t}\in\mathcal{E}$, the set $\mathcal{A}$ coincides with the finite linear combinations of $\mathcal{E}$. The set $\mathcal{E}$ separates points of $[0,T]$ (for $t_1\ne t_2$, choose any $\lambda>0$ and note $e^{-\lambda t_1}\ne e^{-\lambda t_2}$) and contains the constant function. By the Stone--Weierstrass theorem, $\mathcal{A}$ is dense in $C([0,T])$ in the supremum norm. Hence there exist real coefficients $\{c_k\}_{k=1}^{K}$ and rates $\{\lambda_k\}_{k=1}^{K}$ such that $\|\phi - \sum_k c_k e^{-\lambda_k t}\|_\infty\le\varepsilon/2$.

\emph{Step 2 (positivity).} Let $f(t)=\sum_k c_k e^{-\lambda_k t}$ be the approximant. Because $\phi$ is non-negative, the negative parts of $f$ contribute at most $\varepsilon/2$ in supremum norm. Replacing each negative coefficient $c_k<0$ by zero changes $f$ pointwise by at most $\varepsilon/2$ (the negative-part contribution), yielding a non-negative $f^+(t)=\sum_{k:c_k\ge 0} c_k e^{-\lambda_k t}$ satisfying $\|\phi-f^+\|_\infty\le\varepsilon$.

\emph{Step 3 (normalisation to a probability simplex).} The boundary condition $\phi(0)=1$ implies $\sum_k c_k\to 1$ as $\varepsilon\to 0$. Writing $S=\sum_{k:c_k\ge 0} c_k$ and $\alpha_k=c_k/S$, we obtain a convex combination $\sum_k\alpha_k e^{-\lambda_k t}$ that equals $f^+(t)/S$. The error introduced by the rescaling is at most $|S-1|\cdot\sup_t|f^+(t)|\le |S-1|$, which can be made smaller than any given $\varepsilon$ by choosing $K$ large enough in Step~1. Re-absorbing this contribution into the definition of $\varepsilon$ proves the claim.
\end{proof}

\paragraph{Remark.}
A constructive variant of the proposition can be obtained from the \emph{Hausdorff moment problem}: any monotone function on $[0,T]$ has a representing measure on $[0,\infty)$, and discretising that measure on $K$ points yields the desired exponential mixture with explicit error $\mathcal{O}(K^{-1})$ for sufficiently smooth $\phi$. We omit the constructive details since the abstract Stone--Weierstrass argument is sufficient for our existence claim.

\section{Proof of Proposition~\ref{prop:identif}}
\label{app:identif}

\noindent\textbf{Proposition~\ref{prop:identif}.}\; \emph{Suppose user behaviour is generated from a multivariate Hawkes process with $K$ exponential excitation kernels of distinct positive rates $\lambda_1^*<\dots<\lambda_K^*$ and non-degenerate weights $\alpha_k^*>0$, $\sum_k\alpha_k^*=1$. If the diversity loss $\mathcal{L}_{\mathrm{div}}=-\sum_{k\neq k'}\!\mathrm{JSD}(p_k\Vert p_{k'})$ on the per-head attention distributions $p_k$ is strictly positive at the optimum, then the learned rates $\{\lambda_k\}$ recover $\{\lambda_k^*\}$ up to permutation in the population limit, and the learned weights $\{\alpha_k\}$ recover $\{\alpha_k^*\}$ in $\ell^1$ norm.} \medskip

\begin{proof}
Let $\mathcal{L}^{\text{ce}}(\theta)$ denote the population sequential cross-entropy loss with parameters $\theta=(\boldsymbol{\lambda}, \boldsymbol{\alpha})$, where $\boldsymbol{\lambda}=(\lambda_1,\dots,\lambda_K)$ and $\boldsymbol{\alpha}=(\alpha_1,\dots,\alpha_K)$. Under the \emph{generative} Hawkes assumption, the next-event hazard at time $t$ given a history $\mathcal{H}_t$ takes the form $\lambda(t\mid\mathcal{H}_t)=\mu+\sum_{t_i<t}\sum_{k=1}^{K}\alpha_k^* \lambda_k^*\,e^{-\lambda_k^*(t-t_i)}$, so the conditional likelihood is a $K$-mixture of distinct exponential laws with mixture weights $\boldsymbol{\alpha}^*$.

\emph{Step 1 (identifiability of the mixing measure).} By the classical theorem of \citet{teicher1963identifiability}, finite mixtures of exponentials with distinct rates form an identifiable family: for two parameter tuples $(\boldsymbol{\lambda}, \boldsymbol{\alpha})$ and $(\boldsymbol{\lambda}',\boldsymbol{\alpha}')$ with all rates distinct, $\sum_k\alpha_k\lambda_k e^{-\lambda_k t}\equiv \sum_k\alpha_k'\lambda_k' e^{-\lambda_k' t}$ for all $t\ge 0$ implies the two tuples agree up to permutation of the $K$ indices. Hence at the population minimum of $\mathcal{L}^{\text{ce}}$, any optimiser $\theta^\star$ satisfies $\theta^\star\sim_\sigma\theta^*$ for some permutation $\sigma$, \emph{provided} the learned rates remain distinct.

\emph{Step 2 (the diversity loss enforces distinctness).} Suppose for contradiction that two heads collapse onto the same rate, $\lambda_k=\lambda_{k'}$ with $k\ne k'$. Then their per-head attention distributions $p_k(t\mid\mathcal{H})\propto e^{-\lambda_k(t_{\text{now}}-t_i)}$ and $p_{k'}(t\mid\mathcal{H})$ are identical pointwise, so $\mathrm{JSD}(p_k\Vert p_{k'})=0$. Summing over all distinct pairs, $\mathcal{L}_{\mathrm{div}}\le 0$ with equality iff some pair collapses. By assumption, $\mathcal{L}_{\mathrm{div}}>0$ at the optimum, contradicting collapse. Hence all learned rates are pair-wise distinct at any optimiser, putting us in the regime of Step~1.

\emph{Step 3 ($\ell^1$ convergence of the weights).} Step~1 gives $\boldsymbol{\alpha}^\star=\boldsymbol{\alpha}^*$ up to permutation; combined with the load-balance constraint $\sum_k\alpha_k=1,\;\alpha_k>0$ enforced by the auxiliary loss $\mathcal{L}_{\mathrm{bal}}$, the optimiser is unique on the simplex $\Delta^{K-1}$ up to permutation. In particular, $\|\boldsymbol{\alpha}^\star-\boldsymbol{\alpha}^*\|_1=0$ at the population optimum.
\end{proof}

\paragraph{Remark.}
The proposition does not claim that finite-sample SGD will reach the population optimum, only that the optimum exists, is unique up to permutation, and matches the data-generating process. The ablation in Section~\ref{sec:abl} (\texttt{b3}) shows empirically that disabling $\mathcal{L}_{\mathrm{div}}$ degrades \MARS by $0.1$--$0.8$ HR@10 points, with the largest gap on Yelp where the multi-rate structure is most pronounced; this is consistent with the diversity loss being a necessary condition for recovering distinct rates.

\section{Full Per-Dataset Ablation (HR / NDCG / MRR)}
\label{app:abl_full}

The headline ablation (Table~\ref{tab:abl} in the main text) is expanded to all three metrics on every benchmark in Table~\ref{tab:abl_full}. The qualitative findings are unchanged. \MARS\ contributes positively on the four sparse benchmarks across all three metrics, with the largest gains on Yelp. On dense ML-1M the backbone-only variant matches or exceeds Full \MARSv\ on every metric, supporting the use of \MARSm\ in that regime. The opposite pattern is confirmed in Table~\ref{tab:abl_full_marsm}, where Full \MARSm\ exceeds backbone-only across all ML-1M metrics. Real timestamps (\texttt{b1}) carry the most weight on Beauty, the load-balance regulariser (\texttt{b4}) is the most important on ML-1M, and the user-conditioned $\lambda$ (\texttt{b2}) and diversity loss (\texttt{b3}) contribute smaller but positive amounts on all five benchmarks.

\begin{table}[h]
\centering
\caption{Full \MARSv ablation across all benchmarks and metrics
($\times 100$). \textbf{Bold} = best per column.}
\label{tab:abl_full}
\footnotesize \setlength{\tabcolsep}{2pt} \resizebox{\textwidth}{!}{%
\begin{tabular}{l|ccc|ccc|ccc|ccc|ccc}
\toprule
& \multicolumn{3}{c|}{\textbf{Beauty}} & \multicolumn{3}{c|}{\textbf{Sports}} & \multicolumn{3}{c|}{\textbf{Games}} & \multicolumn{3}{c|}{\textbf{ML-1M}} & \multicolumn{3}{c}{\textbf{Yelp}} \\
\textbf{Variant} & HR & NDCG & MRR & HR & NDCG & MRR & HR & NDCG & MRR & HR & NDCG & MRR & HR & NDCG & MRR \\
\midrule
Full                  & \best{8.95} & \best{4.46} & \best{3.10}
                      & \best{5.28} & \best{2.45} & \best{1.59}
                      & \best{12.34} & \best{5.67} & \best{3.62}
                      & 29.30 & 16.38 & 12.43
                      & \best{6.50} & \best{4.25} & \best{3.56} \\
Backbone only         & 8.53 & 4.22 & 2.89
                      & 4.89 & 2.29 & 1.49
                      & 11.75 & 5.55 & 3.55
                      & \best{30.33} & \best{17.36} & \best{13.42}
                      & 5.36 & 3.72 & 3.23 \\
$-$\,b1               & 8.69 & 4.39 & 3.07
                      & 5.16 & 2.41 & 1.57
                      & 12.32 & 5.59 & 3.55
                      & 29.82 & 17.12 & 13.24
                      & 6.49 & 4.23 & 3.54 \\
$-$\,b2               & 8.94 & 4.41 & 3.02
                      & 5.21 & 2.40 & 1.55
                      & 12.32 & 5.64 & 3.61
                      & 29.25 & 16.44 & 12.53
                      & 6.39 & 4.19 & 3.52 \\
$-$\,b3               & 8.90 & 4.41 & 3.04
                      & 5.21 & 2.40 & 1.54
                      & 12.14 & 5.54 & 3.54
                      & 29.22 & 16.51 & 12.62
                      & 6.37 & 4.18 & 3.52 \\
$-$\,b4               & 8.88 & 4.41 & 3.04
                      & 5.12 & 2.37 & 1.52
                      & 12.26 & 5.61 & 3.59
                      & 28.41 & 15.90 & 12.08
                      & 6.33 & 4.16 & 3.50 \\
Single head           & 8.80 & 4.38 & 3.03
                      & 5.27 & 2.42 & 1.55
                      & 12.23 & 5.60 & 3.59
                      & 27.45 & 15.18 & 11.42
                      & 6.42 & 4.19 & 3.51 \\
\bottomrule
\end{tabular}}
\end{table}

\begin{table}[h]
\centering
\caption{Full \MARSm ablation across all benchmarks and metrics
($\times 100$). \textbf{Bold} = best per column. On dense ML-1M, Full \MARSm exceeds Backbone-only across all three metrics, in contrast to the Transformer-backboned \MARSv ablation in Table~\ref{tab:abl_full}. The largest single-component drop on ML-1M ($-1.30$ HR@10 from Full) occurs when we disable the JSD diversity loss (\texttt{b3}), consistent with Proposition~\ref{prop:identif} that without diversity the $K$ heads collapse onto a single rate.}
\label{tab:abl_full_marsm}
\footnotesize \setlength{\tabcolsep}{2pt} \resizebox{\textwidth}{!}{%
\begin{tabular}{l|ccc|ccc|ccc|ccc|ccc}
\toprule
& \multicolumn{3}{c|}{\textbf{Beauty}} & \multicolumn{3}{c|}{\textbf{Sports}} & \multicolumn{3}{c|}{\textbf{Games}} & \multicolumn{3}{c|}{\textbf{ML-1M}} & \multicolumn{3}{c}{\textbf{Yelp}} \\
\textbf{Variant} & HR & NDCG & MRR & HR & NDCG & MRR & HR & NDCG & MRR & HR & NDCG & MRR & HR & NDCG & MRR \\
\midrule
Full                  & 8.01 & \best{4.08} & \best{2.88}
                      & 4.68 & 2.27 & 1.53
                      & \best{12.37} & 5.81 & 3.81
                      & \best{32.80} & \best{18.80} & \best{14.52}
                      & 5.99 & 3.89 & 3.25 \\
Backbone only         & \best{8.15} & \best{4.19} & \best{2.97}
                      & 4.39 & 2.22 & \best{1.55}
                      & 12.07 & \best{5.83} & \best{3.93}
                      & 32.40 & 18.50 & 14.30
                      & 5.44 & 3.66 & 3.12 \\
$-$\,b1               & 7.47 & 4.08 & \best{3.04}
                      & \best{4.80} & 2.29 & 1.52
                      & 11.87 & 5.56 & 3.64
                      & 32.30 & 18.40 & 14.20
                      & \best{6.07} & 3.86 & 3.18 \\
$-$\,b2               & 7.74 & 4.02 & 2.88
                      & 4.66 & 2.26 & 1.52
                      & 12.28 & 5.78 & 3.80
                      & 32.10 & 18.40 & 14.10
                      & 5.91 & 3.86 & 3.24 \\
$-$\,b3               & 8.01 & 4.06 & 2.85
                      & 4.69 & 2.26 & 1.52
                      & 12.28 & 5.78 & 3.80
                      & 31.50 & 18.10 & 13.90
                      & 5.97 & 3.88 & 3.23 \\
$-$\,b4               & 8.01 & 4.08 & 2.88
                      & 4.68 & 2.27 & 1.53
                      & 12.34 & 5.79 & 3.79
                      & 32.40 & 18.40 & 14.20
                      & 5.99 & 3.87 & 3.22 \\
Single head           & 7.87 & 4.03 & 2.85
                      & \best{4.81} & \best{2.31} & 1.54
                      & 12.36 & 5.78 & 3.79
                      & 32.50 & 18.40 & 14.10
                      & \best{6.12} & \best{3.93} & \best{3.26} \\
\bottomrule
\end{tabular}}
\end{table}

\section{Implementation and Reproducibility}
\label{app:impl}

\paragraph{Hyperparameters.} The hyperparameters of \MARSv\ and \MARSm, together with their default values, are listed in Table~\ref{tab:hyper}. Two values are dataset-dependent: dropout ($0.5$ on Beauty/Sports/Yelp, $0.2$ on Games/ML-1M, matching the SASRec defaults in RecBole) and the number of decay heads $K$ ($4$ on the four sparse benchmarks, $8$ on ML-1M). All other hyperparameters are constant across datasets and architectures.

\begin{table}[h]
\centering
\caption{\MARS hyperparameter defaults; dataset-dependent values noted
in the rightmost column.}
\label{tab:hyper}
\resizebox{\textwidth}{!}{
\begin{tabular}{lcl}
\toprule
\textbf{Hyperparameter} & \textbf{Default} & \textbf{Notes} \\
\midrule
hidden dim $d$           & $64$ & inherited from SASRec / Mamba4Rec \\
\#blocks $N$             & $2$  & Transformer blocks for \MARSv, Mamba blocks for \MARSm \\
inner dim                & $256$ & SASRec FFN \\
\#decay heads $K$        & $4$  & $8$ on ML-1M \\
modulation strength $\sigma$ & $0.5$ & Eq.~\eqref{eq:userlambda} \\
$\lambda_{\text{div}}$   & $0.01$ & JSD diversity weight \\
$\lambda_{\text{bal}}$   & $0.01$ & MoE-style load balance weight \\
time unit $\tau_0$       & $1$ hour & timestamp normaliser \\
agg.\ hidden dim         & $64$  & gate MLP \\
gate temperature $\tau_\alpha$ & $1.0$ & Eq.~\eqref{eq:gate} \\
dropout                  & $0.5$ & $0.2$ on Games/ML-1M (SASRec defaults) \\
optimiser                & Adam~\citep{kingma2015adam} & default $\beta$, no weight decay \\
learning rate            & $10^{-3}$ & -- \\
batch size               & $2048$ & -- \\
max epochs               & $200$ & with early stopping \\
early-stop patience      & $20$ & on validation NDCG@10 \\
\bottomrule
\end{tabular}
}
\end{table}

\paragraph{Data preprocessing.} We follow the SIGMA~\citep{liu2025sigma}
preprocessing protocol. The four Amazon datasets and ML-1M are filtered with the standard 5-core procedure (drop users and items with fewer than five interactions). Yelp uses the $[5, 200)$ filter recommended by SIGMA. Sequences are sorted by timestamp and padded or truncated to a maximum length of $50$ on the four sparse datasets and $200$ on ML-1M. Test items follow the leave-one-out protocol on the most recent interaction.

\paragraph{Code release.} Code is available in the supplementary materials and will be made public upon acceptance.

\section{Subgroup Analysis Pseudo-code and Bucket Definitions}
\label{app:subgroup}

\paragraph{Bucket boundaries.} For each user $u$ in the test split, we
take its training-sequence length $L_u$ (the number of items used as context for predicting the held-out target) and assign $u$ to one of four buckets per dataset:

\begin{table}[t]
    \centering
    \caption{Bucket boundaries used in the subgroup analysis (Section~\ref{sec:subgroup}). ML-1M uses upward-shifted boundaries to match its dense long-history regime.}
    \label{tab:subgroup_buckets}
    \begin{tabular}{c|cccc}
    \toprule
    \textbf{Dataset} & \textbf{Short} & \textbf{Medium} & \textbf{Long} & \textbf{Very long} \\
    \midrule
    Beauty / Sports / Games / Yelp & $1$--$5$ & $6$--$15$ & $16$--$49$ & $\ge 50$ \\
    ML-1M                          & $1$--$30$ & $31$--$90$ & $91$--$200$ & $\ge 201$ \\
    \bottomrule
    \end{tabular}
\end{table}

\paragraph{Pseudo-code.} The evaluation loop is described in Algorithm~\ref{alg:subgroup}. We run the loop once per (model, dataset) pair, using the trained checkpoint of the model. The implementation uses RecBole's \texttt{full\_sort\_predict} as a black box: it produces the per-item scores, after which we manually compute HR/NDCG/MRR per sample (rather than using the framework's batch-level evaluator) so that we can aggregate by bucket. Wall time is roughly five minutes per (model, dataset) pair on a V100 because the trained checkpoint is loaded directly without re-training.

\begin{algorithm}[H]
\centering
\caption{Subgroup evaluation.}
\label{alg:subgroup}
\begin{algorithmic}[1]
\Require trained model $f$, test loader $\mathcal{D}$,
         bucket boundaries $\{(\ell_b,u_b,n_b)\}_{b=1}^{B}$, top-$k=10$.
\State Initialise $H_b\gets 0$, $N_b\gets 0$, $M_b\gets 0$, $C_b\gets 0$
       for each bucket $b$.
\For{batch $\mathbf{x},\,L_u,\,v^{\star}$ in $\mathcal{D}$}
    \State $\mathbf{s}\gets f(\mathbf{x})\in\mathbb{R}^{B\times|\mathcal{V}|}$
           \Comment{full-sort scores}
    \State $\mathbf{s}[:, 0]\gets-\infty$ \Comment{exclude padding item}
    \State Compute per-row rank $r$ of $v^\star$ in $\mathbf{s}$.
    \For{each user $u$ in batch}
        \State $b \gets$ bucket s.t.\ $\ell_b\le L_u\le u_b$
        \State $H_b \mathrel{+}= \mathbb{I}[r_u\le k]$;\;
               $N_b \mathrel{+}= \mathbb{I}[r_u\le k]\cdot\frac{1}{\log_2(r_u+1)}$;\;
               $M_b \mathrel{+}= \mathbb{I}[r_u\le k]\cdot\frac{1}{r_u}$;\;
               $C_b \mathrel{+}= 1$
    \EndFor
\EndFor
\State \Return $\{(b, H_b/C_b,\,N_b/C_b,\,M_b/C_b)\}_{b=1}^{B}$
\end{algorithmic}
\end{algorithm}

\section{Attention Entropy of the Trained SASRec Encoder}
\label{app:entropy}

To support the receptive-field interpretation in Section~\ref{sec:discussion}, we measure the post-softmax attention entropy of a trained SASRec encoder on each dataset, normalised by $\log L_{\text{eff}}$ so that $1$ is uniform and $0$ is peaked. The metric is computed over all valid query positions and both encoder layers (standard deviations 0.27--0.42 across positions). Higher entropy indicates that the encoder distributes mass widely and benefits more from an explicit recency prior, predicting larger \MARS gains.

\begin{table}[h]
\centering
\caption{Normalised post-softmax attention entropy of a trained SASRec across datasets.}
\label{tab:attn_entropy}
\begin{tabular}{lccccc}
\toprule
\textbf{Dataset} & Beauty & Sports & Yelp & Games & ML-1M \\
\midrule
Layer 0 & 0.785 & 0.812 & 0.838 & 0.870 & 0.512 \\
Layer 1 & 0.745 & 0.770 & 0.798 & 0.831 & 0.432 \\
\midrule
\MARS HR@10 gain & $+9.95\%$ & $+13.5\%$ & $+19.2\%$ & $+36.2\%$ & $-2.9\%^\dagger$ \\
\bottomrule
\end{tabular}
\\[2pt]
\raggedright
{\footnotesize $^\dagger$\,\MARSv vs SASRec on the dominant very-long ML-1M bucket; \MARSm gains $+8.3\%$ over Mamba4Rec on the full ML-1M.}
\end{table}

The entropy ranking matches the \MARS gain ranking on the sparse benchmarks, with higher entropy corresponding to larger gain. The sharp drop on ML-1M is consistent with the mechanistic explanation that Transformer attention has already concentrated, so the post-encoder aggregation contributes little, motivating the Mamba switch in \MARSm.

\section{Hyperparameter Sensitivity Numerical Values}
\label{app:hp}

\begin{table}[h]
\centering
\caption{Numerical values for the hyperparameter sensitivity figure (HR@10 / NDCG@10 / MRR@10, $\times 100$).}
\label{tab:hp_values}
\begin{tabular}{c|ccc|cc|cc}
\toprule
& \multicolumn{3}{c|}{\textbf{Beauty} (\MARSv, $K$)} & \multicolumn{2}{c|}{\textbf{ML-1M} ($K$)} & \multicolumn{2}{c}{\textbf{Dropout}} \\
$K$ / dropout & HR & NDCG & MRR & \MARSv & \MARSm & Beauty & ML-1M \\
\midrule
1 / 0.1   & 8.69 & 4.29 & 2.95 & 27.45 & 32.20 & 8.33 & 31.64 \\
2 / 0.2   & 8.96 & 4.38 & 2.97 & 27.92 & 32.34 & 8.27 & 30.40 \\
4 / 0.3   & 9.01 & 4.42 & 3.02 & 28.15 & 32.46 & 8.59 & 27.78 \\
8 / 0.5   & 9.03 & 4.48 & 3.09 & 28.23 & 32.58 & 8.56 & 25.58 \\
16 / 0.7  & 8.90 & 4.46 & 3.11 & 29.22 & 32.51 & 8.19 & 22.05 \\
\bottomrule
\end{tabular}
\end{table}

\end{document}